\documentclass[%
reprint,
 amsmath,amssymb,
prb,
]{revtex4-1}

\usepackage{physics}% load physics package
\usepackage{graphicx}% Include figure files
\usepackage{dcolumn}% Align table columns on decimal point
\usepackage{bm}% bold math
\usepackage{hyperref}% add hypertext capabilities
\usepackage[utf8]{inputenc}

\usepackage{amsthm}
\usepackage{tikz}
\usepackage{cleveref}
\usetikzlibrary{arrows.meta, calc, decorations.pathreplacing}

\newtheorem{proposition}{Proposition}

\begin{document}

%\preprint{APS/123-QED}

\title{Classification with Quantum Measurements}% Force line breaks with \\
%\thanks{A footnote to the article title}%

\author{Fabio A. González}
\email{fagonzalezo@unal.edu.co}

\affiliation{MindLab Research Group, Departamento de Ingeniería de Sistemas e Industrial\\
 Universidad Nacional de Colombia, Bogotá, Colombia }
\author{Vladimir Vargas-Calderón}
%\email{vvargasc@unal.edu.co }
\author{Herbert Vinck-Posada}

\affiliation{Grupo de Superconductividad y Nanotecnología, Departamento de Física\\
 Universidad Nacional de Colombia, Bogotá, Colombia}

\date{\today}% It is always \today, today,
             %  but any date may be explicitly specified

\begin{abstract}
This paper reports a novel method for supervised machine learning based on the mathematical formalism that supports quantum mechanics. The method uses projective quantum measurement as a way of building a prediction function. Specifically, the relationship between input and output variables is represented as the state of a bipartite quantum system. The state is estimated from training samples through an averaging process that produces a density matrix. Prediction of the label for a new sample is made by performing a projective measurement on the bipartite system with an operator, prepared from the new input sample, and applying a partial trace to obtain the state of the subsystem representing the output. The method can be seen as a generalization of Bayesian inference classification and as a type of kernel-based learning method. One remarkable characteristic of the method is that it does not require learning any parameters through optimization. We illustrate the method with different 2-D classification benchmark problems and different quantum information encodings.
\end{abstract}

%Use showkeys class option if keyword
                              %display desired
\maketitle

\section{Introduction}

In recent years, the interest in the combination of quantum information processing and machine learning has been growing fueled by the increasing popularity and advances in both fields \cite{Biamonte2017QuantumLearning}. The field product of the intersection of these research areas is commonly denoted as quantum machine learning \cite{Schuld2015AnLearning}. The new field has produced a considerable amount of research work that explores different interactions between the two areas \cite{Perdomo-Ortiz2018OpportunitiesComputers}. 

The different approaches to quantum machine learning can be broadly classified into four categories depending on whether a classical or quantum system generates the data and whether the processing device is a classical computer or a quantum computer \cite{Schuld2018SupervisedComputers}. In the category of classical-data/quantum-processing, a large amount of work has been devoted to the development of quantum versions of different classical machine learning algorithms with an emphasis on showing an advantage, at least theoretically, of the quantum version in terms of speedup \cite{Schuld2018SupervisedComputers}. The classical-data/classical-processing category refers to the use of tools of quantum information research to formulate machine learning methods that take advantage of the quantum conceptual machinery. This category has been less explored than the former one and is the primary motivation of the work discussed in this paper.
%A representative approach in this category is the use of tensor networks, a tool for efficient modeling and simulation of many-body quantum systems in hybrid learning methods that combine them with neural networks. 

This paper presents a classification method based on the mathematical formalism that supports quantum mechanics. The method can be implemented both as an algorithm for a classical computers and as a hybrid classical/quantum algorithm. The main idea of the method is to represent the joint probability of input and output variables, $P(x,y)$, as the state of a bipartite quantum system. Training corresponds to calculating this state from training samples. Prediction corresponds to performing a projective measurement with an operator, prepared from the new input sample to be classified, and subsequently calculating a partial trace to obtain the state of the output subsystem.

The representation of $P(x,y)$ as the state of a quantum system, more specifically as a density matrix, generalizes the classical probabilistic representation and enriches it with the additional representation capabilities of the quantum formalism. We show that the method generalizes Bayesian inference and can also be seen as a type of kernel learning method. Another remarkable feature of the proposed framework is that training, unlike many machine learning methods, does not require to optimize a cost function that depends on parameters. Instead, training is done merely by averaging quantum states representing training samples.

Different works have addressed the implementation of supervised learning models based on formalism from quantum mechanics or quantum information processing. \citet{Lloyd2013QuantumLearning} present a quantum algorithm for supervised cluster assignment based on calculating the minimum distance from a sample to the centroids representing clusters. Analogous quantum algorithms based on nearest-neighbor classification have been explored by \citet{Wiebe2015QuantumLearning, Sergioli2018AClassifier, Schuld2017ImplementingCircuit} among others. Different quantum versions of classical machine learning algorithms have been studied: support vector machines \citep{Anguita2003QuantumMachines, Rebentrost2014QuantumClassification, Li2015ExperimentalMachine}, decision trees~\citep{Lu2014QuantumClassifier}, classifier ensembles~\citep{Schuld2018QuantumClassifiers}, neural networks~\citep{Schuld2015SimulatingComputer, Schuld2014TheNetwork, Wiebe2016QuantumModels, Wan2017QuantumNetworks}, among others. Another line of work is the application of methods traditionally used for modeling quantum systems to supervised machine learning. Tensor networks, a tool for efficient modeling and simulation of many-body quantum systems, are the most representative of these methods~\citep{Stoudenmire2016SupervisedNetworks,Stoudenmire2018LearningNetworks} and have been applied to different classification problems including image classification~\citep{Klus2019Tensor-basedClassification}, language analysis~\citep{ Zhang2019ASpace} and probabilistic modeling~\citep{Glasser2019ExpressiveModeling}. With the exception of nearest-neighbor-based methods, all the other quantum machine learning methods rely on optimization to learn the parameters of the model. The method presented in this paper does not make use of optimization since learning is accomplished by averaging quantum states. The method can be seen as a form of kernel-based learning, but in contrast with typical kernel methods and nearest-neighbor learning, there is no need for storing any individual training sample to be used during prediction.

\section{Quantum Measurement Classification (QMC)}

The proposed method is similar in principle to a generative Bayesian inference approach \cite{ng2002discriminative}. Generative supervised learning models estimate the joint probability of inputs and outputs $P(x, y)$ that is used during the prediction stage to calculate the conditional probability of the output given a new sample. During training, QMC estimates the joint probability of inputs and outputs from training samples and represents it as a density matrix, $\rho_{\text{train}}$, that corresponds to a quantum state of a bipartite system $S_{\mathcal{XY}}=S_\mathcal{X}+S_\mathcal{Y}$. $S_\mathcal{X}$ is the subsystem representing the inputs with associated Hilbert space $\mathcal{H}_{\mathcal{X}}$ and the subsystem $S_\mathcal{Y}$ represents the outputs in the Hilbert space $\mathcal{H}_{\mathcal{Y}}$. Consequently, the representation space of the system $S_\mathcal{XY}$ is $\mathcal{H}_{\mathcal{X}} \otimes \mathcal{H}_{\mathcal{Y}}$. Prediction is made by performing a quantum measurement over $S_\mathcal{XY}$ with an operator specifically prepared from a new input sample.

\begin{figure}[!ht]
    \centering
    \includegraphics[width=\columnwidth]{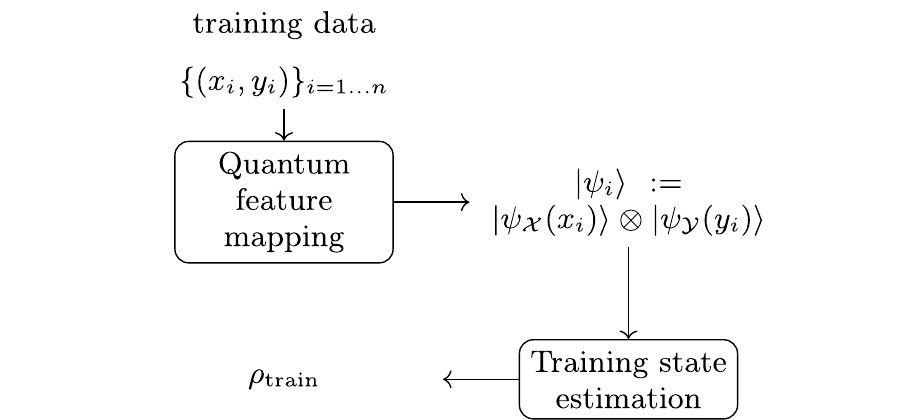}
    \caption{Training process: training samples are represented as quantum states of a bipartite system; states are averaged to calculate the training density matrix. }
    \label{fig:training}
\end{figure}

\Cref{fig:training} shows the training process that consists of two main steps, quantum feature mapping, and training state estimation:
\begin{enumerate}
    \item Quantum feature mapping. In this step each training sample is  mapped to $\mathcal{H}_{\mathcal{X}} \otimes \mathcal{H}_{\mathcal{Y}}$ using the following function:
    \begin{align}
    \begin{aligned}
        \psi: {\mathcal{X}} \times {\mathcal{Y}}& \to \mathcal{H}_{\mathcal{X}} \otimes \mathcal{H}_{\mathcal{Y}}  \\
         (x,y) &\mapsto \ket{\psi_{\mathcal{X}}(x)} \otimes \ket{\psi_{\mathcal{Y}}(y)},
    \end{aligned}
    \end{align}
    where $\psi_{\mathcal{X}}:{\mathcal{X}}\rightarrow \mathcal{H}_{\mathcal{X}}$ and $\psi_{\mathcal{Y}}:{\mathcal{Y}}\rightarrow \mathcal{H}_{\mathcal{Y}}$ are functions that map inputs and outputs, respectively, to the corresponding quantum Hilbert spaces. As a short-hand notation, every data sample $(x_i, y_i)\in T$ is mapped to the quantum feature space as $\psi:(x_i,y_i)\mapsto\ket{\psi_i}:= \ket{\psi_{\mathcal{X}}(x_i)}\otimes \ket{\psi_{\mathcal{Y}}(y_i)}$. Here, $T$ is a set of $n$ training samples.
    
    \item Training state estimation. In this step, we calculate a quantum state that summarizes the training data set. This state is represented by a density matrix $\rho_{\text{train}} $ There are three alternatives to calculate $\rho_{\text{train}}$:
    \begin{itemize}
        \item Pure state. In this case the training state corresponds to a superposition of the states representing training samples. First we calculate the superposition state $\ket{\psi_{\text{train}}}$ as:
        \begin{equation}\label{eq:psi-train-pure}
        \ket{\psi_{\text{train}}} = 
        \frac{\sum_{i=1}^n \ket{\psi_i}}{\norm{\sum_{i=1}^n \ket{\psi_i}}},    
        \end{equation}
        and define 
        \begin{equation}\label{eq:rho-train-pure}
        \rho_{\text{train}} = \ket{\psi_{\text{train}}}\bra{\psi_{\text{train}}}.    
        \end{equation}

        \item Mixed state. Here, $\rho_{\text{train}}$ corresponds to a mixture of the states corresponding to the training samples:
        \begin{equation}\label{eq:rho-train-mixed}
        \rho_{\text{train}} = 
        \frac{1}{n} \sum_{i=1}^n \ket{\psi_i} \bra{\psi_i}.
        \end{equation}
        
        \item Classical mixture. In this case we extract the probabilities associated with the quantum state $\ket{\psi_i}$ and use them to build a quantum state, represented by a density matrix, that only has classical uncertainty:
        
        \begin{equation}\label{eq:rho-train-classic}
        \rho_{\text{train}} = 
        \frac{1}{n} \sum_{i=1}^n \sum_{j=1}^m \abs{\braket{\psi_i}{j}}^2 \ket{j}\bra{j},
        \end{equation}
        
        where $m$ is the dimension of $\mathcal{H}_{\mathcal{X}} \otimes \mathcal{H}_{\mathcal{Y}}$ and $\ket{j}$ are the elements of the canonical basis.
    \end{itemize}
\end{enumerate}

The three alternatives to calculate the training state in step 2 correspond to three different ways of combining quantum and classical uncertainty in a quantum state \cite{jacobs2014quantum}. The pure state (\cref{eq:rho-train-pure}) encodes the training data set using only quantum uncertainty, the classical mixture (\cref{eq:rho-train-classic}) encodes the training samples using classical probabilities exclusively, and the mixed state (\cref{eq:rho-train-mixed}) uses a combination of quantum and classical uncertainty to encode the training samples in the training quantum state.

\begin{figure}[!ht]
    \centering
    \includegraphics[width=\columnwidth]{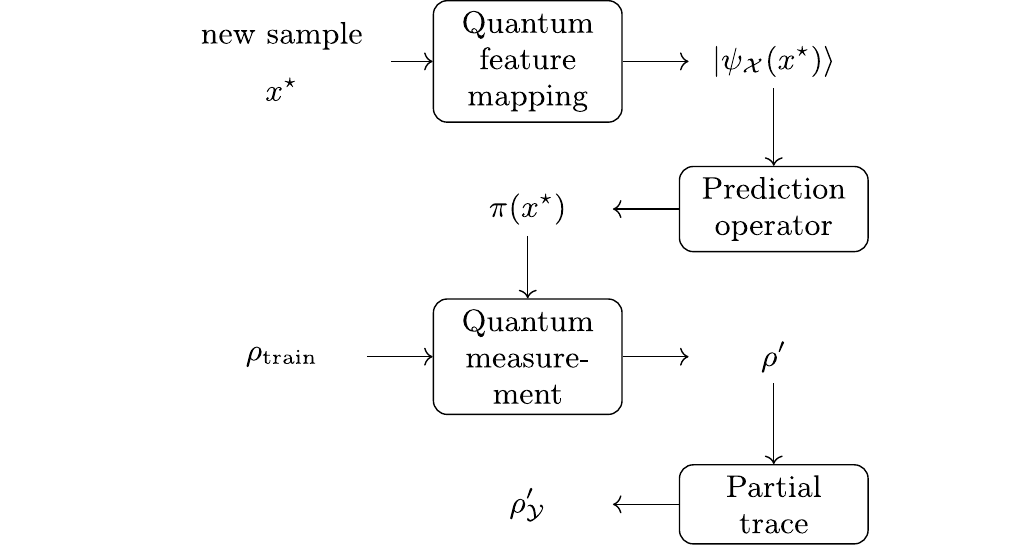}
    \caption{Prediction process: a new sample is represented as a quantum state; a projective measure operator is built from the input quantum state; the operator is applied to the training density matrix, a partial trace is used to calculate a density matrix of the $y$ subsystem, which represent the prediction.}
    \label{fig:prediction}
\end{figure}

The prediction process is depicted in \cref{fig:prediction}. The process receives as input a new sample $x^\star$ to be classified, and the training state $\rho_{\text{train}}$ from the training phase. The steps of the prediction process are described next.
\begin{enumerate}
    \item Quantum feature mapping. $x^\star$ is mapped to $\ket{\psi_{\mathcal{X}}(x^\star)}$.
    
    \item Prediction operator.  An operator acting on $\mathcal{H}_{\mathcal{X}} \otimes \mathcal{H}_{\mathcal{Y}}$ is defined as follows:
    \begin{equation}\label{eq:pred-operator}
    \pi(x^\star) = \ket{\psi_{\mathcal{X}}(x^\star)} \bra{\psi_{\mathcal{X}}(x^\star)} \otimes \text{Id}_{\mathcal{H}_{\mathcal{Y}}},
    \end{equation}
    where $\text{Id}_{\mathcal{H}_{\mathcal{Y}}}$ is the identity operator on $\mathcal{H}_{\mathcal{Y}}$. 

    \item Quantum measurement. The operator $\pi(x^\star)$ is applied to $\rho_{\text{train}}$:
    \begin{equation}\label{eq:measurement}
    \rho' = \frac{\pi(x^\star)\rho_{\text{train}}\pi(x^\star)}
    {\Tr[\pi(x^\star)\rho_{\text{train}}\pi(x^\star)]}
    \end{equation}
    
    \item Partial trace. The partial trace of $\rho'$ with respect to subsystem $S_{\mathcal{X}}$ is calculated. It corresponds to the reduced state of $\rho'$ on subsystem $S_{\mathcal{Y}}$:
    \begin{equation}\label{eq:partial-trace}
    \rho'_{{\mathcal{Y}}}=\Tr_{\mathcal{X}}[\rho']
    \end{equation}
\end{enumerate}

The density matrix $\rho'_{{\mathcal{Y}}}$ contains information about the state of the subsystem $S_{\mathcal{Y}}$ after the state of the subsystem $S_{\mathcal{X}}$ is projected onto the state $\ket{\psi_{\mathcal{X}}(x^\star)}$. This density matrix gives information about the probability of predictions. For instance if ${\mathcal{Y}}= \{y_k\}_{k=1\dots m}$ and $\psi_{\mathcal{Y}}$ corresponds to a one-hot or a probability encoding (see Subsection \ref{subsect:one-hot-encoding}), the diagonal element $\rho'_{{\mathcal{Y}}i,i}$ can be interpreted as the probability of the value $y_i$, i.e., the probability that $x^\star$ is labeled as $y_i$.

QMC not only resembles generative Bayesian inference, but it also generalizes it. The following proposition formally states this claim.

\begin{proposition}{}
\label{prop:bayesian}
Let $T=\{(x_i, y_i)\}_{i=1,\dots,n}$ be a set of training samples, $x^\star$ a sample to classify, with $x_i,x^\star \in \{1, \dots m\}$ and $y_i \in \{1, 2\}$. Let $\rho_{\text{train}}$ be the state calculated using the mixed state, \cref{eq:rho-train-mixed} or equivalently the classical mixture \cref{eq:rho-train-classic}, and a one-hot encoding feature map for both $x_i$ and $y_i$. Then the diagonal elements of the density matrix $\rho'_{{\mathcal{Y}}}$ calculated using \cref{eq:partial-trace} correspond to an estimation, using Bayesian inference, of the conditional probabilities $P(y=i|x^\star)$:

\begin{equation}
    \rho'_{{\mathcal{Y}}i,i} = P(y=i|x^\star) = \frac{P(x^\star|y=i)P(y=i)}{P(x^\star)},
\end{equation}

where $P(x^\star|y=i)$, $P(y=i)$ and $P(x^\star)$ are estimated from $T$.
\end{proposition}

\begin{proof}
Since both $x_i$ and $y_i$ are represented using a one-hot encoding representation, then
\begin{align}
\ket{\psi_i} = \ket{x_i} \otimes \ket{y_i} = \ket{x_i y_i}.  
\end{align}

Applying \cref{eq:rho-train-mixed}:
\begin{align}\label{eq:rho-train-prop1}
\begin{aligned}
    \rho_{\text{train}} ={} &  \frac{1}{n} \sum_{i=1}^n \ket{x_i y_i} \bra{x_i y_i}  \\
    {}={} & \sum_{j=1}^m \sum_{k=1}^2 P(x=j, y=k) \ket{jk}\bra{jk},
\end{aligned}
\end{align}
with $
P(x=j, y=k) = \frac{1}{n}\sum_{i=1}^n \delta_{x_i j}  \delta_{y_i k}
$. Applying the prediction operator (\cref{eq:pred-operator})
\begin{align}
\pi(x^\star) = \ket{x^\star} \bra{x^\star} \otimes \text{Id}_{H_\mathcal{Y}}
\end{align}
to \cref{eq:rho-train-prop1} produces
\begin{align}\label{eq:rho-prime-prop1}
\begin{aligned}
    \rho' & = \frac{\sum_{k=1}^2 P(x=x^\star, y=k) \ket{x^\star k}\bra{x^\star k}}
    {\sum_{k=1}^2 P(x=x^\star, y=k)} \\
    & = \sum_{k=1}^2 P(y=k | x=x^\star) \ket{x^\star k}\bra{x^\star k}.
\end{aligned}
\end{align}
Finally, we calculate the partial trace of \cref{eq:rho-prime-prop1} to obtain:
\begin{align}
\begin{aligned}
    \rho'_{\mathcal{Y}} & = \text{Tr}_{\mathcal{X}} [\rho'] \\
    & = \sum_{k=1}^2 P(y=k | x=x^\star) \ket{k}\bra{k}.
\end{aligned}
\end{align}

\end{proof}

Using a one-hot encoding makes the general mixture (\cref{eq:rho-train-mixed}) equivalent to a classical mixture (\cref{eq:rho-train-classic}). Proposition \ref{prop:bayesian} states that  estimating the training state using a classical mixture is equivalent to do classical Bayesian inference. This is not surprising since the classical mixture corresponds to a conventional probabilistic encoding of the information in the training data set. 

When using the more general quantum feature map along with a mixed state (\cref{eq:rho-train-mixed}) to estimate the training quantum state, the prediction process involves more complex interactions between states. The following proposition shows that, in this case, the resulting density matrix $\rho'_{{\mathcal{Y}}}$ for the subsystem $S_{\mathcal{Y}}$ corresponds to a linear combination of the density matrices representing the output variables of the training samples.

\begin{proposition}{}
Let $T=\{(x_i, y_i)\}_{i=1\dots n}$ be a set of training samples, $x^\star$ a sample to classify, with $x_i,x^\star \in {\mathcal{X}}$ and $y_i \in {\mathcal{Y}}$. Let $\rho_{\text{train}}$ be the state calculated using a mixed state (\cref{eq:rho-train-mixed}) and quantum feature maps $\psi_{\mathcal{X}}$ and $\psi_{\mathcal{Y}}$. Then the density matrix $\rho'_{{\mathcal{Y}}}$, calculated with \cref{eq:partial-trace}, can be expressed as: 

\begin{equation}\label{eq:kernel-equivalence2}
    \rho'_{{\mathcal{Y}}} = \mathcal{M}\sum_{i=1}^n \abs{k(x^\star,x_i)}^2  \ket{\psi_{\mathcal{Y}}(y_i)} 
    \bra{\psi_{\mathcal{Y}}(y_i)},
\end{equation}

where $k(x^\star,x_i) = \braket{\psi_{\mathcal{X}}(x^\star)}{\psi_{\mathcal{X}}(x_i)}$ and $\mathcal{M}^{-1} = \Tr[\pi(x^\star)\rho_{\text{train}}\pi(x^\star)]$ .
\end{proposition}

\begin{proof}
\Cref{eq:rho-train-mixed} can be expressed as:

\begin{align}\label{eq:rho-train-prop2}
\begin{aligned}
    \rho_{\text{train}} & =
    \frac{1}{n} \sum_{i=1}^n \ket{\psi_i} \bra{\psi_i} \\
    & = \frac{1}{n} \sum_{i=1}^n \ket{\psi_{\mathcal{X}}(x_i)} \bra{\psi_{\mathcal{X}}(x_i)} \otimes \ket{\psi_{\mathcal{Y}}(y_i)} \bra{\psi_{\mathcal{Y}}(y_i)} \\
    & \equiv \frac{1}{n} \sum_{i=1}^n \sigma_{\mathcal{X}}(x_i) \otimes \sigma_{\mathcal{Y}}(y_i),
\end{aligned}
\end{align}
where $\sigma_{\mathcal{X}}(x_i) = \ket{\psi_{\mathcal{X}}(x_i)} \bra{\psi_{\mathcal{X}}(x_i)}$ and $\sigma_{\mathcal{Y}}(y_i) = \ket{\psi_{\mathcal{Y}}(y_i)} \bra{\psi_{\mathcal{Y}}(y_i)}$. 
Applying \cref{eq:measurement} to \cref{eq:rho-train-prop2} we get:

\begin{align}\label{eq:rho-prime-prop2}
\begin{aligned}
    \rho' & = \mathcal{M} \sum_{i=1}^n \sigma_{\mathcal{X}}(x^\star) \sigma_{\mathcal{X}}(x_i) \sigma_{\mathcal{X}}(x^\star) \otimes \sigma_{\mathcal{Y}}(y_i) \\
    & = \mathcal{M} \sum_{i=1}^n \abs{k(x^\star, x_i)}^2 \sigma_{\mathcal{X}}(x^\star) \otimes \sigma_{\mathcal{Y}}(y_i),
\end{aligned}
\end{align}
where $k(x^\star,x_i) = \braket{\psi_{\mathcal{X}}(x^\star)}{\psi_{\mathcal{X}}(x_i)}$ and $\mathcal{M}^{-1} = n \Tr[\pi(x^\star)\rho_{\text{train}}\pi(x^\star)]$. 

Finally, we calculate the partial trace of \cref{eq:rho-prime-prop2} to obtain:
\begin{align}
\begin{aligned}
    \rho'_{\mathcal{Y}} & = \text{Tr}_{\mathcal{X}} [\rho'] \\
    & = \mathcal{M} \sum_{i=1}^n k(x^\star, x_i)^2  \sigma_{\mathcal{Y}}(y_i) \\
    & = \mathcal{M}\sum_{i=1}^n k(x^\star,x_i)^2  \ket{\psi_{\mathcal{Y}}(y_i)} 
    \bra{\psi_{\mathcal{Y}}(y_i)}
\end{aligned}
\end{align}

\end{proof}

\Cref{eq:kernel-equivalence2} can be seen as type of kernel-based classification function $f(x) =  \sum_{(x_i, y_i) \in S} \alpha_i k(x, x_i) y_i$, where $k$ is a kernel function and the $\alpha_i$ are learned parameters. In QMC $\alpha_i:=k^*(x^\star, x_i)$ and $y_i$ is replaced by $\ket{\psi_{\mathcal{Y}}(y_i)} \bra{\psi_{\mathcal{Y}}(y_i)}$. Notice that $k(x^\star, x_i)$ corresponds to the dot product in the quantum Hilbert space $H_{\mathcal{X}}$, so it is in fact a kernel function. This means that QMC can be seen as type of kernel-based learning method. However an important difference is that while conventional kernel methods require to learn, usually through optimization, the $\alpha_i$ parameters, in QMC there are not parameters to be learned. Nevertheless, the method has hyperparameters, mainly associated to feature maps to be discussed in next section, that has to be estimated or fixed by the user. This is done through cross validation, the common practice in machine learning.

It is important to note that it is possible to use \cref{eq:kernel-equivalence2} to calculate $\rho'_y$, however this will require keeping all the training samples. QMC does not require this as training samples are compactly represented by the $\rho_\text{train}$ density matrix.

It is worth emphasizing that QMC can, in principle, be implemented in quantum devices through the preparation of a pure training state of the form \cref{eq:psi-train-pure} with well-known preparation protocols~\citep{Plesch2011Quantumstate,shende2006synthesis}. With the same protocol, the state of the new data sample $x^\star$ can be built. Finally, the projective measurement can be achieved via a third ancillary state that allows a SWAP test~\citep{Buhrman2001Fingerprinting,Garcia-Escartin2013SwapTest,Cincio_2018,PARK2020126422}, as in other distance-based classifiers~\citep{blank2019quantum,Schuld2017iop} (with the upside that data need not be radially separable, as well as keeping the size of the physical system to be independent from the amount of training data).

Assuming that the dimensions of $\mathcal{H}_{\mathcal{X}}$ and $\mathcal{H}_{\mathcal{Y}}$ are $k$ and $\ell$ respectively, a straighforward implementation of QMC training takes time $O(k\ell n)$. This means that training is linear on the training set size. Storing $\rho_{\text{train}}$ requires space $O((k\ell)^2)$. For prediction, the most costly process is calculating \cref{eq:measurement}. A direct implementation will take time $O((k\ell)^4)$, however a more careful implementation can perform it in time $O((k\ell)^3)$. This can be further reduced by performing a previous eigendecomposition of $\rho_{\text{train}}$.

\section{Quantum feature maps}
In quantum machine learning literature, there are several approaches to represent data as quantum states. \citet{Schuld2018SupervisedComputers} propose different strategies such as basis encoding, encoding data directly as qubits, and amplitude encoding, which is encoding data in the amplitude of quantum states~\citep{Schuld2017iop}. Next, we describe several approaches that we use in the illustrative examples.

\subsection{Softmax encoding}
A common approach to represent real numbers in the interval $[0,1]$ is to use the mapping $\phi:x \mapsto \sin{(2\pi x)}\ket{0} + \cos{(2\pi x)}\ket{1}$, encoding the number as a the superposed state of a qubit. We propose a softmax quantum encoding that extends this approach from two dimensions to multiple dimensions.

First we define a probability encoding for real values $P:\mathbb{R}\rightarrow [0,1]^m$ where $m$ is the number of states:
\begin{align} 
    P_i(x)=\left(\frac{\exp{-\beta\|x-\alpha_i\|^2}}{\sum_{j=1}^m \exp{-\beta\|x-\alpha_j\|^2}}\right)_{i=1\dots m},
\end{align}
where $\alpha_i = \frac{i-1}{m-1}$.
Using these probabilities we build a quantum state representing a real number
\begin{align}
\ket{\varphi(x)} = \sum_{j=1}^m \sqrt{P_i(x)}\ket{j}.
\end{align}
The quantum state corresponding to a vector $x_i \in \mathbb{R}^n$ is defined as
\begin{align}
\ket{\psi_{\mathcal{X}}(x_i)} = \ket{\varphi(x_{i,1})} \otimes \dots \otimes \ket{\varphi(x_{i,n})}.
\end{align}

\subsection{One-hot encoding} \label{subsect:one-hot-encoding}
This representation corresponds to a basis encoding for discrete variables with $m$ possible values, $\mathcal{X}=\{1, \dots, m\}$. The encoding for $x=j$ is given by
\begin{align}
\psi_{\mathcal{X}}(j) = \ket{j}. 
\end{align}

\subsection{Squeezed states}

Recently, \citet{Schuld2019QuantumSpaces} proposed to encode data to the phase of a light squeezed state
\begin{align}
    \ket{(r,\varphi)} = \frac{1}{\sqrt{\cosh(r)}}\sum_{n=0}^\infty \frac{\sqrt{(2n)!}}{2^n n!}(e^{i(\varphi + \pi)}\tanh(r))^n\ket{2n},
\end{align}
so that a vector $x_i\in [0, \pi]^n$ is mapped to $\psi_{\mathcal{X}}(c, x_i) = \ket{(c, x_{i,1})}\otimes \ldots \otimes \ket{(c, x_{i,n})}$.

\subsection{Coherent states}

Data can also be encoded into the average number of photons of a canonical coherent state \cite{Chatterjee2017coherent}:

\begin{align}
    \ket{(\alpha, \gamma)} = e^{-\frac{\gamma|\alpha|^2}{2}}\sum_{n=0}^\infty \frac{\alpha^n\gamma^{n/2}}{\sqrt{n!}}\ket{n}\label{eq:coherent_state}
\end{align}

where a scaling characterized by $\gamma$ has been introduced so that the dot product of the two states corresponds to a Gaussian kernel with $\gamma$ parameter. The mapping from a real data sample $x_j\in\mathbb{R}^n$ to the complex $\alpha$ is done as follows. An auxiliary variable $\theta_j$ is built through a min-max scaling of the data set to $[0,\pi]^n$, so that $x_{j,\ell}\mapsto x_{j,\ell} e^{i\theta_{j,\ell}}$. Therefore, a data point $x_j$ is mapped to the quantum feature space through
\begin{align}
    \psi_{\mathcal{X}}(x_j,\gamma) = \ket{(x_{j,1}e^{i\theta_{j,1}},\gamma)}\otimes\ldots\otimes \ket{(x_{j,n}e^{i\theta_{j,n}},\gamma)}
\end{align}
which induces a kernel
\begin{align}
    \abs{k_{\gamma}(x_k, x_j)}^2 = \prod_{\ell=1}^n \exp(-\gamma|x_{k,\ell}e^{i\theta_{k,\ell}} - x_{j,\ell}e^{i\theta_{j,\ell}} |^2),
\end{align}
where the argument of the exponential is explicitly $-\gamma(x_{k,\ell}^2 + x_{j,\ell}^2 - 2 x_{k,\ell} x_{j,\ell} \cos(\theta_{k,\ell} - \theta_{j,\ell}))$, which imposes a higher distance penalty in the feature space for distant data points in the original space $\mathbb{R}^n$ than the usual Gaussian kernel.

\subsection{Random Fourier Features}
As in quantum state representations, the feature space of kernel methods is a Hilbert space. This means that a quantum feature map implicitly defines a kernel. A natural question is whether the opposite conversion also works, i.e., given a particular kernel function, can we find a quantum feature map such that the inner product of the corresponding Hilbert space corresponds to the kernel. In general, the answer is no; however, it is possible to find an approximation for certain kernels. Random Fourier features (RFF) \cite{rahimi2008random} provides a technique that finds an explicit Hilbert space such that the inner product in this space approximates a shift-invariant kernel. Specifically, for a given kernel $k:\mathbb{R}^d \times \mathbb{R}^d \rightarrow \mathbb{R}$, RFF finds a map $z:\mathbb{R}^d \rightarrow \mathbb{R}^D$ such that $k(x,y) \approx z'(x)z(y)$.

The quantum state corresponding to a vector $x_i \in \mathbb{R}^d$ is defined as
\begin{align}
\ket{\psi_{\mathcal{X}}(x_i)} = \frac{1}{||z(x_i)||}\sum_{j=1}^D z_j(x_i)\ket{j}.
\end{align}

\section{Method illustration}

\begin{figure}[!ht]
    \centering
    \includegraphics[width=\columnwidth]{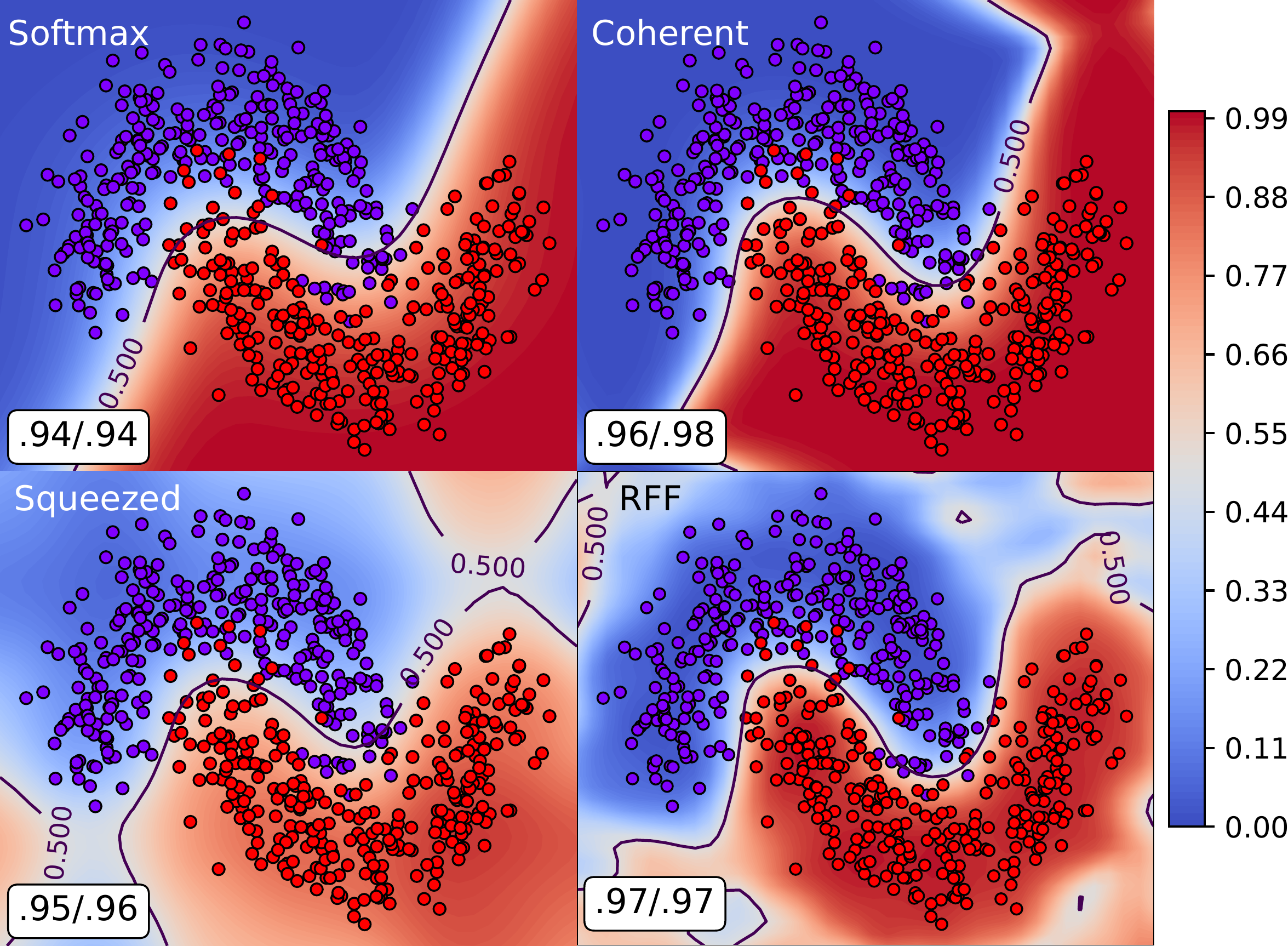}
    \caption{Decision heatmap for a two-moons dataset of the softmax, coherent, squeezed, and random Fourier features states-based feature maps truncated at 20 Fock states. The regularization parameters were $\beta=70$ for the softmax state, $\gamma=70$ for the coherent state, $r=2.5$ for the squeezed state, and $\gamma=20$ for the random Fourier features state. In all four cases, a mixed training state was used. The white boxes show the train/test accuracy of the classifier. The values of the parameters were tuned using cross-validation.}

    \label{fig:comparison_mixed}
\end{figure}

In this section, we illustrate the performance of QMC with binary classification toy problems for the different aforementioned feature maps. \Cref{fig:comparison_mixed} compares the decision boundary obtained from a mixed training state through the different feature maps, where the states are truncated up to the first 20 Fock states for each input feature. The color tells the probability that the output state belongs to the red or blue classes. For all the cases, the method achieves high discrimination in both classes: 94\%, 98\%, 96\% and 97\% accuracy in the test set for the softmax, coherent, squeezed and RFF encodings, respectively.

\begin{figure}[!hb]
    \centering
    \includegraphics[width=\columnwidth]{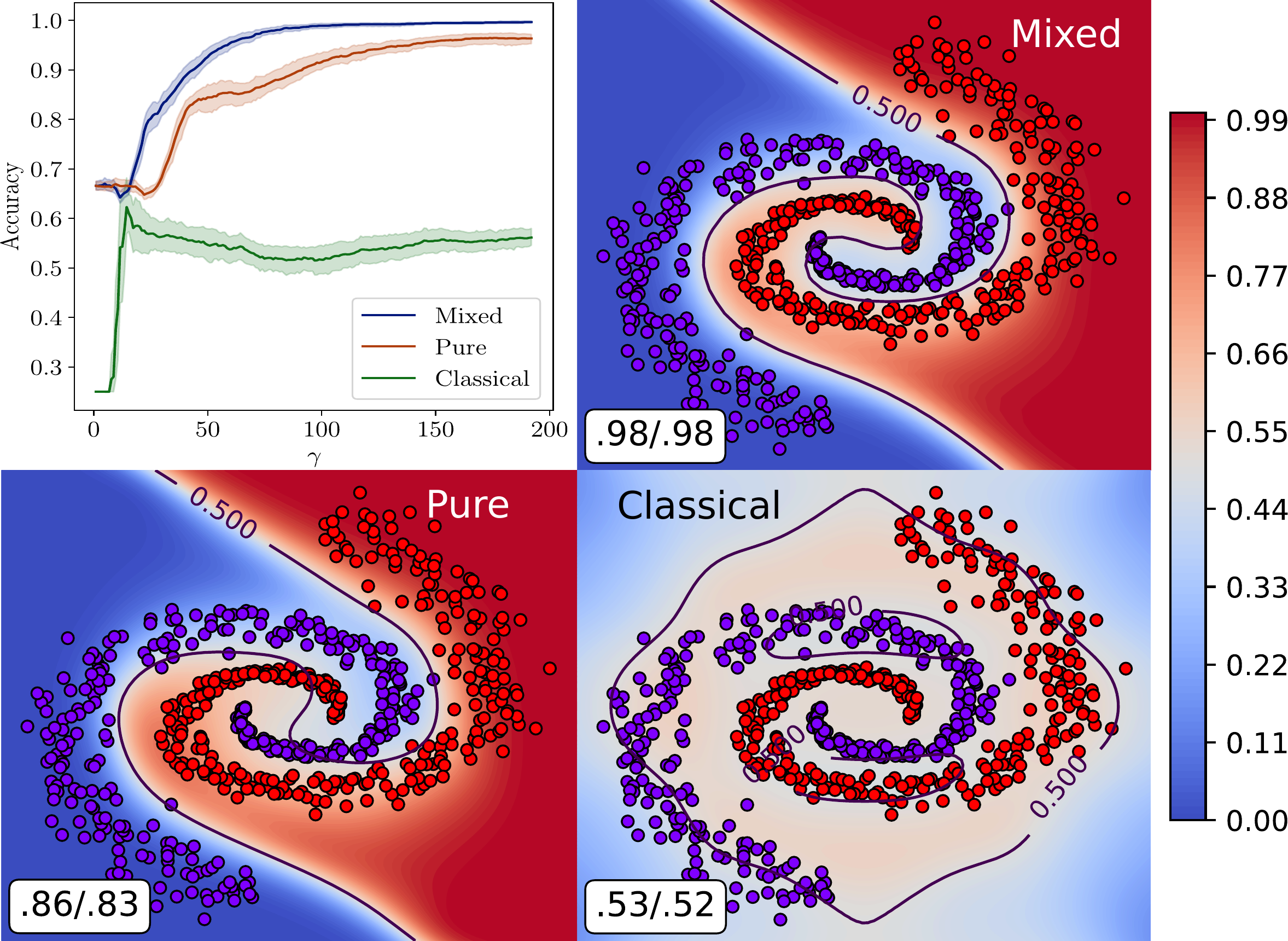}
    \caption{Classification accuracy as a function of the regularization parameter $\gamma$ on a spirals dataset for mixed, pure and classical training states built with the coherent state encoding truncated at 32 Fock states. Decision heatmaps are shown for the three training schemes at $\gamma=70$. The white boxes are as in \cref{fig:comparison_mixed}.}
    \label{fig:coherent_spirals}
\end{figure}

\Cref{eq:rho-train-pure,eq:rho-train-mixed,eq:rho-train-classic} correspond to three different alternatives to estimate the training state. The mixed and pure alternatives are expected to take advantage of the quantum correlations induced by the feature mapping and exploited in the projective measurement process. \Cref{fig:coherent_spirals} shows that this is, in fact, the case for the 2D spirals dataset. The three plots show prediction regions for the three different estimation strategies using the coherent state quantum feature mapping truncated to the first 32 Fock states. The mixed state representation has the best performance (98\% accuracy in the test set), followed closely by the pure state representation (83\% accuracy in the test set). Both are able to capture the particular shape of both classes. The classical state representation fails to do good discrimination. This is better observed on the top-left plot where the classification precision is measured for a range of $\gamma$ values\footnote{The values of $\gamma$, $r$ and $\beta$ have a similar effect on classification accuracy: the larger the value, the larger the discrimination. This means that regions are assigned large probabilities of belonging to one class or the other. This property is good for separable data, but can lead to overfitting for non-separable data.}, showing that the worst classification scheme is the classical representation state, whereas the best one is the mixed representation state, closely followed by the pure representation state. The same behavior is observed when the coherent state is truncated to the first 20--64 Fock states (not shown).

Regarding the squeezed state encoding, decision boundaries for a circles dataset for pure and mixed training states are shown in \cref{fig:squeezed_comparison}. Here, the mixed training state outperforms the pure training state. The classical training state is useless because the data is mapped to the phase of the squeezed state, and the probabilities in \cref{eq:rho-train-classic} do not depend on this phase. Again, a signature of the kernel induced by the squeezed state is seen in the regions at the middle top, bottom, right, and left parts of the decision heatmaps that are wrongly classified. These regions emerge from the fact that the similarity induced by the squeezed-based kernel between two points at a fixed Euclidean distance is maximum if the horizontal or vertical components of the two points are the same.

\begin{figure}[!t]
    \centering
    \includegraphics[width=\columnwidth]{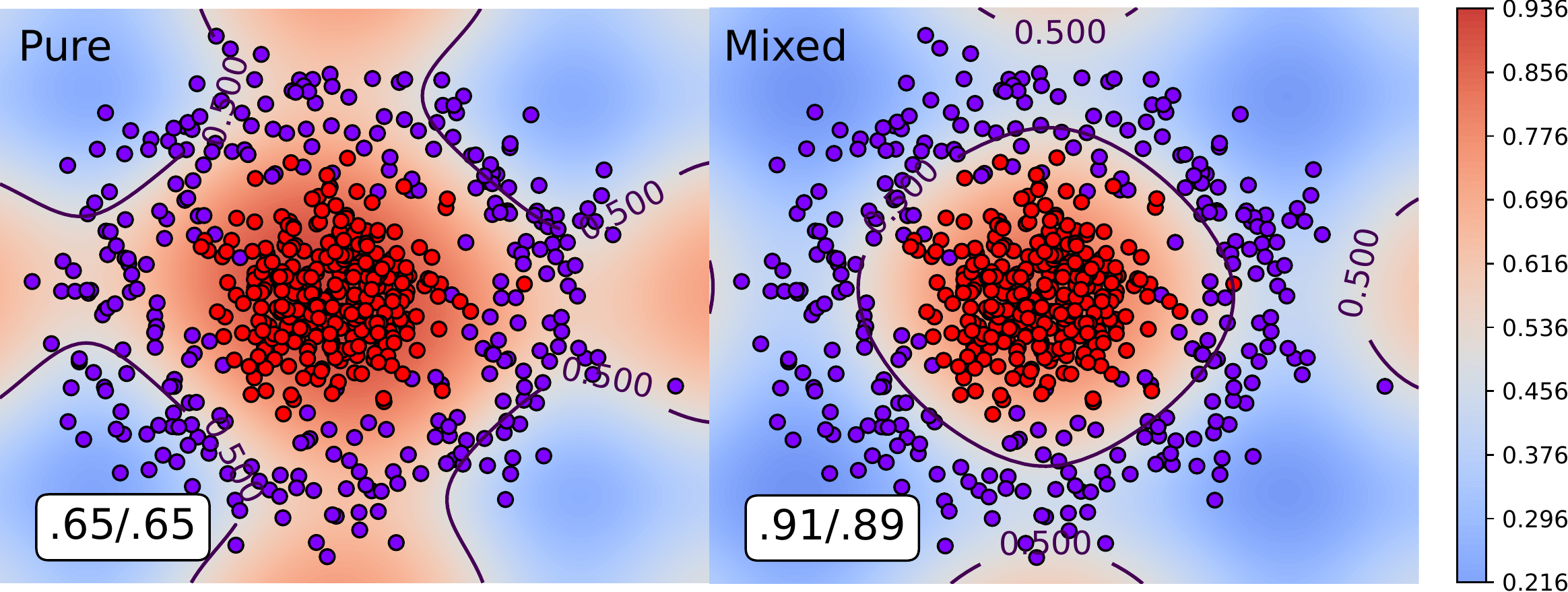}
    \caption{Decision heatmap for a two-circles dataset of the squeezed state feature map with pure and mixed training states. Squeezed states were truncated to the first 10 contributing Fock states, and a value of $r=2.5$ was used, as in \citep{Schuld2019QuantumSpaces}. The white boxes are as in \cref{fig:comparison_mixed}.}
    \label{fig:squeezed_comparison}
\end{figure}

The accuracies obtained by every quantum feature map on each of the discussed toy-datasets are shown in~\cref{tab:all_results}. The coherent quantum feature map shows similar results on all datasets, but random features and softmax quantum feature maps show lower performance in the spirals dataset. Interestingly, the squeezed state, whose related kernel is anisotropic~\citep{Schuld2019QuantumSpaces} performs much better in the spirals dataset.

\begin{table}[!hb]
\centering
\caption{Train/test accuracy for every quantum feature map and toy-dataset discussed in this paper for mixed training states. The number of Fock states considered for each dataset was 10 for circles, 20 for moons and 32 for spirals. The specific parameters for each quantum feature map were $r=2.5$ for squeezed states, $\beta=70$ for softmax states, $\gamma=20$ for random feature map states, and $\gamma=70$ for coherent states.}
\label{tab:all_results}
\begin{tabular}{lllll}
\toprule
QFM &   Coherent &        RFF &    Softmax &   Squeezed \\
Dataset &            &            &            &            \\
\colrule
Circles &  0.96/0.94 &  0.88/0.87 &  0.94/0.93 &  0.91/0.89 \\
Moons   &  0.96/0.98 &  0.95/0.97 &  0.94/0.94 &  0.95/0.96 \\
Spirals &  0.98/0.98 &  0.82/0.75 &  0.85/0.83 &   1.00/0.99 \\
\botrule
\end{tabular}
\end{table}

\section{Conclusions}

This paper presented a classification method based on quantum measurement. The overall strategy of the method is based on two mechanisms: first, to represent the joint probability of inputs and outputs by the state of a bipartite quantum system and, second, to predict the outputs of new input samples performing a quantum measurement. 

Using this quantum measurement framework as a basis for function induction contributes a two-fold novel perspective to supervised quantum machine learning. On the one hand, the training process does not require optimization of parameters, since training corresponds to state averaging. This is an essential departure from current machine learning models, both classical and quantum-based. On the other hand, the classification model induced by QMC can be understood as a generalization of Bayesian-inference classification and as a type of kernel classification model. This connection is a consequence of the harmonious combination of linear algebra and probability provided by the quantum framework. Some works connect kernel and probabilistic methods \cite{jaakkola1999exploiting, muandet2017kernel}; however the quantum measurement framework constitutes a novel unifying perspective.

The ability of QMC of inducing a classification model without parameter optimization suggests the possibility of an efficient classical implementation. This is the case for the training process, whose time complexity is linear on the number of training samples. However, the computational burden moves from the training process to the prediction process and from time complexity to space complexity. In particular the space required by the training density matrix, $\rho_\text{train}$, is $\mathcal{O}(m^2\ell^2)$ where $m=\abs{\mathcal{H}_\mathcal{X}}$ and $\ell=\abs{\mathcal{H}_\mathcal{Y}}$. Scaling QMC to large scale learning problems requires dealing with this space complexity. A promising research line to address this problem is to use tensor networks~\citep{Stoudenmire2016SupervisedNetworks} to build a compact representation of these density matrices employing tensor factorizations. Another possibility of mitigating the computational costs is the implementation of QMC as a hybrid classical-quantum algorithm. Whether this could exploit a quantum advantage is part of our future research.

\bibliography{references,refs}% Produces the bibliography via BibTeX.

\end{document}